\documentclass[onecolumn,aps,12pt]{revtex4-1}

\usepackage{amsmath}

\begin{document}

\newtheorem{theorem}{Theorem}[section]
\newtheorem{lemma}[theorem]{Lemma}
\newtheorem{proposition}[theorem]{Proposition}
\newtheorem{corollary}[theorem]{Corollary}

\newenvironment{proof}[1][Proof]{\begin{trivlist}
\item[\hskip \labelsep {\bfseries #1}]}{\end{trivlist}}
\newenvironment{definition}[1][Definition]{\begin{trivlist}
\item[\hskip \labelsep {\bfseries #1}]}{\end{trivlist}}
\newenvironment{example}[1][Example]{\begin{trivlist}
\item[\hskip \labelsep {\bfseries #1}]}{\end{trivlist}}
\newenvironment{remark}[1][Remark]{\begin{trivlist}
\item[\hskip \labelsep {\bfseries #1}]}{\end{trivlist}}

\newcommand{\qed}{\nobreak \ifvmode \relax \else
      \ifdim\lastskip<1.5em \hskip-\lastskip
      \hskip1.5em plus0em minus0.5em \fi \nobreak
      \vrule height0.75em width0.5em depth0.25em\fi}

\date{\today}

\title{Entropy in the Kuramoto model and its implications for the stability of partially synchronized states}

\author{Anders Nordenfelt}

\affiliation{Departamento de F\'{i}sica, Universidad Rey Juan Carlos, Tulip\'{a}n s/n,  28933 M\'{o}stoles, Madrid, Spain}

\begin{abstract}
We discuss the concept of entropy applied to the infinite-N Kuramoto model and derive an expression for its time derivative. The time derivative of the entropy functional is shown to depend on the synchronization order parameter in a very simple way and, absent diffusion, it is never increasing. The implications of this for the stability of partially synchronized states is discussed. We conclude with a section on the entropy of the marginal density function averaged over all natural frequencies.
\end{abstract}

\maketitle

\section{introduction}

Since it was first introduced in the 1970s, the Kuramoto model \cite{Kuramoto1984, Strogatz2000} has been a continuous source of inspiration for developments in various fields of nonlinear dynamics. The initial motivation behind the model was to elucidate the phenomenon of collective synchronization observed on many occasions in nature, for example the synchronous chorusing of crickets. As stated in its original form, the model consists of $N$ phase variables $\theta_i$ that are coupled sinusoidally all-to-all according to the following equation:
\begin{equation}\label{original}
\dot{\theta}_i(t) = \omega_i + \frac{K}{N}\sum_{j=1}^{N}\sin[\theta_j(t) - \theta_i(t)].
\end{equation}    
The natural frequencies $\omega_i$ of each oscillator are drawn from some distribution $g(\omega)$ that is usually assumed to be symmetric around zero. In order to quantify the rate of synchronization the order parameter $r$ is introduced, which in the case a finite system is given implicitly by the equation
\begin{equation}\label{finitesynch}
re^{i\psi} = \frac{1}{N}\sum_{j=1}^{N} e^{i\theta_j(t)}.
\end{equation}
The state variable $\psi$ occuring on the left hand side of Eq. ($\ref{finitesynch}$) could be seen as the collective phase of the entire ensemble of oscillators. With this definition it is possible to rewrite the system of equations ($\ref{original}$) in the following form
\begin{equation}\label{original2}
\dot{\theta}_i(t) = \omega_i + Kr\sin[\psi - \theta_i(t)].
\end{equation} 
One of the main results of Kuramoto's early analysis was the derivation of a formula for the coupling strength $K_c$ critical for the onset of collective synchronization:
\begin{equation}
K_c = \frac{2}{\pi g(0)}.
\end{equation}
Furthermore, if we assume a Lorenzian shape of the frequency distribution $g(\omega) = \frac{\gamma}{\pi(\gamma^2 + \omega^2)}$ the calculations can be continued to obtain also a formula for the resulting order parameter
\begin{equation}
r = \sqrt{1 - \frac{K_c}{K}}.
\end{equation}     
Subsequently an infinite-N version of the Kuramoto model was formulated by Sakaguchi \cite{Sakaguchi1988} and independently by Mirollo and Strogatz \cite{Mirollo1991, Strogatz2000}. In this case the individual oscillators are replaced with a density function $\rho(\theta, t, \omega)$ representing the relative amount of oscillators with natural frequency around $\omega$ having the phase $\theta$ at time $t$. Now, the density function $\rho$ is governed by the continuity equation
\begin{equation}\label{fundamental0}
\frac{\partial \rho}{\partial t} = - \frac{\partial}{\partial \theta}(\rho \nu),
\end{equation}
where the velocities $\nu$ are given by
\begin{equation}\label{velocity}
\nu(\theta, t, \omega) = \omega + Kr\sin(\psi - \theta).  
\end{equation}
In the infinite-N model the synchronization order parameter is given by the integral
\begin{eqnarray}\label{synch}
re^{i\psi} = \int_{-\pi}^{\pi} \int_{-\infty}^{\infty} e^{i\theta} \rho(\theta,t,\omega)g(\omega)\textrm{d}\omega\textrm{d}\theta.
\end{eqnarray}
There are some immediate remarks we can make about equation ($\ref{fundamental0}$) which we state without proof. First of all, the number of oscillators of each frequency is preserved, more specifically: 
\begin{equation}\label{normcons}
\int_{-\pi}^{\pi} \dot{\rho} \textrm{d}\theta = 0
\end{equation}
Secondly, if $\rho$ is initially non-negative everywhere then it will remain non-negative everywhere. Once the infinite-N model has been formulated the question of the system's behaviour around the critical coupling strength can now be studied in another setting. Of particular importance is the uniform incoherent state given by
\begin{equation}
\rho_0(\theta, \omega) \equiv \frac{1}{2\pi}.
\end{equation}  
The question asked was how, for various coupling strengths $K$, the uniform incoherent state would behave under small perturbations. The linear stability analysis carried out by Strogatz and Mirollo  \cite{Mirollo1991, Strogatz2000} corroborated the previous findings that for $K > K_c$ the uniform incoherent state is unstable. However, at the same time it was found that for $K_c < K$ the uniform incoherent state does not go from unstable to asymptotically stable, as one might expect, but instead becomes \emph{neutrally} stable. Similar results were later obtained also for partially synchronized states \cite{Mirollo2007, Mirollo2012}. One of the aims of this paper is to reproduce some variants of these findings by means of the entropy functional, which will be presented in the following section.

\section{Entropy of the density function}\label{sec_entropy}

As pointed out by many people before, the uniform incoherent state is not the only state with zero synchronization order parameter $r$. One of the questions leading up to the material of this paper was: Is there another framework in which the uniform incoherent state is unique? An answer to that question is the Entropy functional $S[\rho]$: 
\begin{equation}\label{entropy}
S[\rho] = -\int \rho \ln(\rho).
\end{equation}
The entropy formula ($\ref{entropy}$) occurs in classical thermodynamics under the name Gibb's entropy, but the same expression can also be found in information theory under the name Shannon entropy. Our main purpose here is not to give the entropy any paticular interpretation although, given the context, the thermodynamic perspective would perhaps be the most natural. Nor will we be worried about the entropy attaining negative values. The first we will do is to confirm that the constant density function indeed maximizes the entropy. This is done by the means of the Euler-Lagrange equations as follows: 
\begin{equation}\label{extremum}
\frac{\partial[\rho_0 \ln(\rho_0)]}{\partial \rho_0} = \ln(\rho_0) + 1 = 0 \Rightarrow \rho_0 \equiv \frac{1}{e}.
\end{equation}
We could have recovered the usual expression $\rho_0 \equiv \frac{1}{2\pi}$ by adding a Lagrange multiplier $\lambda \rho$ in the definition of the entropy, but that would have no effect its time derivative, which is our main concern us here. Thus, in order to avoid unnecessary proliferation of constant factors, for the continuation we will rest content by simply imposing the temporary normalization condition $\int \rho = \frac{2\pi}{e}$. For completeness, at this point we also add a diffusion term to the governing equation which then takes the form
\begin{equation}\label{fundamental}
\frac{\partial \rho}{\partial t} = D\frac{\partial^2 \rho}{\partial \theta^2} - \frac{\partial}{\partial \theta}(\rho \nu).
\end{equation}
To set the stage further, we assume that $\rho$ is a twice continuously differentiable (with respect to $\theta$) non-negative function on the unit circle. The unit circle is identical to the interval $[-\pi, \pi]$ where we have identified the two points $-\pi$ and $\pi$. The periodocity of all functions involved will be very important in the continuation. We define the total entropy as follows:
\begin{equation}\label{totentropy}
S[\rho] = \int_{-\infty}^{\infty} S_{\omega} g(\omega)\textrm{d}\omega
\end{equation}
where
\begin{equation}\label{omegaentropy}
S_{\omega} = -\int_{-\pi}^{\pi} \rho \ln(\rho) \textrm{d}\theta.
\end{equation}
We now proceed to calculate the time-derivative of the entropy:
\begin{equation}\label{step0}
\frac{\mathrm{d}S_{\omega}}{\mathrm{d}t} = -\frac{\mathrm{d}}{\mathrm{d}t}\int_{-\pi}^{\pi} \rho \ln(\rho) \textrm{d}\theta = -\int_{-\pi}^{\pi} \dot{\rho} [\ln(\rho) + 1] \textrm{d}\theta.
\end{equation}
Using Eqns. ($\ref{fundamental}$) and ($\ref{normcons}$) we arrive at
\begin{equation}\label{step1}
\frac{\mathrm{d}S_{\omega}}{\mathrm{d}t} = -\int_{-\pi}^{\pi} \left[D\rho'' - \frac{\partial}{\partial \theta}(\rho \nu)\right] \ln(\rho) \textrm{d}\theta.
\end{equation}\\
The following steps involve integration by parts and the periodicity of all functions repeatedly:
\begin{equation}\label{step2}
\frac{\mathrm{d}S_{\omega}}{\mathrm{d}t} = \int_{-\pi}^{\pi} \left[D\rho' - \rho \nu\right] \frac{\rho'}{\rho} \textrm{d}\theta,
\end{equation}
\begin{equation}\label{step3}
\frac{\mathrm{d}S_{\omega}}{\mathrm{d}t} = D\int_{-\pi}^{\pi} \frac{(\rho')^2}{\rho} \mathrm{d}\theta - \int_{-\pi}^{\pi} [\omega + Kr\sin(\psi - \theta)] \rho' \textrm{d}\theta,
\end{equation}
\begin{equation}\label{step4}
\frac{\mathrm{d}S_{\omega}}{\mathrm{d}t} = D\int_{-\pi}^{\pi} \frac{(\rho')^2}{\rho} \mathrm{d}\theta - \int_{-\pi}^{\pi} Kr\cos(\psi - \theta) \rho \textrm{d}\theta,
\end{equation}\\
Finally, we observe that
\begin{equation}\label{step5}
\cos(\psi - \theta) = \frac{1}{2}(e^{i\psi}e^{-i\theta} + e^{-i\psi}e^{i\theta})
\end{equation}\\
which together with Eq. (\ref{synch}) results in the following remarkably simple expression for the time evolution of the entropy
\begin{equation}\label{entropyderiv}
\frac{\mathrm{d}S[\rho]}{\mathrm{d}t} =  D\int_{-\infty}^{\infty} \int_{-\pi}^{\pi} \frac{(\rho')^2}{\rho} g(\omega) \mathrm{d}\theta \mathrm{d}\omega - Kr^2.
\end{equation}
A similar equation was obtained also in a recent article by Benedetto et al. \cite{Benedetto2014}. We note that, absent diffusion and with positive coupling ($D = 0$, $K > 0$), the entropy is never increasing, which is noteworthy in its own right. Moreover, the distribution of natural frequencies plays no role for the time evolution of the entropy. However, in Sec. $\ref{sec_mean}$ we will see that the natural frequencies play an important role when we consider instead the entropy of the marginal density function averaged over all natural frequencies. 

\section{Implications for the stability of partially synchronized states}

We assume that we have equipped the function space under consideration with some norm $||\rho||$. We now make the following definition:
\begin{definition}
A state $\rho$ is asymptotically stable if there exists an $\epsilon > 0$ such that for any $\rho_1$ satisfying $||\rho_1 - \rho|| < \epsilon$ we have that $||\rho_1 - \rho|| \to 0$ as $t \to \infty$. 
\end{definition}
When considering which norm to choose there is one candidate that comes to mind naturally, namely the $L_1$-norm:
\begin{equation}\label{l1}
||\rho||_{L_1} = \int |\rho|.
\end{equation}
This norm is important for at least for two reasons. First of all, with the $L_1$ norm the governing equation ($\ref{fundamental}$) is norm preserving. Secondly, if $||\rho_1 - \rho||_{L_1} \to 0$ as $t \to \infty$ then the order parameter of $\rho_1$ converges to that of $\rho$. The back-drop, however, is that this norm is too weak to rigorously prove the non-stability conditions by means of the entropy. Therefore, we have reason also to consider the supremum norm:
\begin{equation}\label{supremum}
||\rho||_{\infty} = \textrm{max}|\rho|.
\end{equation}
In the following we will assume the supremum norm unless stated otherwise. We need the following lemmas:
\begin{lemma}\label{neigh}
In any open neighborhood of a state $\rho$ that is either partially synchronized ($0 < r < 1$) or identical to the uniform incoherent state $\rho_0$ there exists a state $\rho_1$ such that $S[\rho_1] < S[\rho]$.
\end{lemma}
\begin{proof}
Proof synopsis: for $\rho = \rho_0$ this holds trivially since $\rho_0$ is the unique state with maximum entropy. For any other state, let $\rho$ propagate for an arbitrarily short time interval under the action of Eq. ($\ref{fundamental}$) with $K > 0$ and $D = 0$. According to Eq. ($\ref{entropyderiv}$) this new state has lower entropy than $\rho$ and can be made to stay within the neighborhood for short enough time interval.
\end{proof}

\begin{lemma}\label{key}
If $||\rho_1 - \rho||_{\infty} \to 0$ as $t \to \infty$ then $S[\rho_1] \to S[\rho]$ as $t \to \infty$.
\end{lemma}
\begin{proof}
Let $Y$ denote the bounded subset of $[-\pi, \pi]$ where $\rho > 1$ and $X$ the bounded subset where $\rho \leq 1$. Pick a sequence $\rho_n$ converging to $\rho$ in the supremum norm. Due to the uniform convergence, eventually all functions $\rho_n$ will attain values greater than $1$ on the set $Y$. For $x,y > 1$ we have that $|y - x| > |\ln(y) - \ln(x)|$. Pick $n$ large enough so that $||\rho_n - \rho||_{\infty} < \epsilon$. Then
\begin{equation}
\begin{split}
|S[\rho_n] - S[\rho]|_{Y} = &|\int_{Y} \rho_n \ln(\rho_n) - \rho \ln(\rho)| = |\int_{Y} (\rho_n - \rho)\ln(\rho_n) - \rho[\ln(\rho) -\ln(\rho_n)] | \leq \\
& \qquad  \int_Y \epsilon\ln(\rho_n) +  \int_Y \epsilon\rho \leq  \int_Y \epsilon(\ln(\rho) + \epsilon) +  \int_Y \epsilon\rho \to 0
\end{split}
\end{equation}
For the set $X$: Due to the continuity of the function $x\ln(x)$ we have that $\rho_n\ln(\rho_n)$ converges to $\rho\ln(\rho)$ pointwise everywhere. Moreover, on the set $X$ the function $\rho_n\ln(\rho_n)$ is eventually bounded in magnitude by the constant $1/e$. These conditions together with the Lebesgue Dominated Covergence Theorem proves the lemma.
\end{proof}
We can now prove the main Theorem, (originally due to Mirollo and Strogatz), by means of the entropy:
\begin{theorem}
Absent diffusion ($D = 0$) and with positive coupling constant $K > 0$ neither $\rho_0$ nor any partially synchronized ($0 < r < 1$) state $\rho$ is asymptotically stable, in the sense of the supremum norm, under the action of Eq. ($\ref{fundamental}$).
\end{theorem}
\begin{proof}
If there was a stable neighborhood around $\rho$, according to Lemma ($\ref{neigh}$) we would be able to find a state $\rho_1$ within this neighborhood such that $S[\rho_1] < S[\rho]$. Furthermore, because of Lemma ($\ref{key}$) we must have that $S[\rho_1] \to S[\rho]$ as $t \to \infty$. But that would imply increasing entropy which is impossible according to Eq. ($\ref{entropyderiv}$). 
\end{proof}
As a final remark, it would seem obvious that Lemma ($\ref{neigh}$), and hence also the rest of the arguments, are valid also for incoherent states ($r = 0$) that are not identical to the uniform incoherent state $\rho_0$. However, we omit the details here.

\section{Entropy of the mean density}\label{sec_mean}

In Section ($\ref{sec_entropy}$) we expressed the time-evolution of the entropy as a function of $r$ but from this little information can be extracted as to how $r$ itself evolves with time. This was one of the main questions posed by Kuramoto. The exponential convergence of the order parameter for certain initial conditions has been proven for example in Ref. \cite{Ott2008}. At first, it might seem counter intuitive that the order parameter converges whilst the density function $\rho$ never finds a stationary state. However, if we consider instead the density function integrated over all natural frequencies
\begin{equation}
\bar{\rho}(\theta) = \int_{-\infty}^{\infty} \rho(\omega, \theta) g(\omega)\textrm{d}\omega,
\end{equation}
henceforth referred to as the \emph{mean density}, the perspective might become more intuitively appealing. We define the entropy of the mean density as follows:
\begin{equation}\label{entropy_mean}
S[\bar{\rho}] = -\int_{-\pi}^{\pi} \bar{\rho}\ln(\bar{\rho}) \textrm{d}\theta. 
\end{equation}
This would be the natural definition of entropy for an observer that is ignorant about the natural frequencies of the oscillators, but more importantly, the mean density function contains sufficient information to calculate $r$:
\begin{equation}
re^{i\psi} = \int_{-\pi}^{\pi} \bar{\rho}(\theta) \textrm{d}\theta.
\end{equation}
Because of the nonlinearity of the logarithm, in general we do not have that $S[\bar{\rho}] = S[\rho]$, but instead we obtain for the time derivative of $S[\bar{\rho}]$ the following expression:
\begin{equation}\label{entropyderiv_mean}
\frac{\mathrm{d}S[\bar{\rho}]}{\mathrm{d}t} = \mathcal{D} + \mathcal{K} + \mathcal{W},
\end{equation}
where we have introduced the notation
\begin{align}
&\mathcal{D} = D\int_{-\pi}^{\pi} \int_{-\infty}^{\infty} \frac{(\rho')^2}{\rho} g(\omega) \mathrm{d}\omega \mathrm{d}\theta, \\
\\
&\mathcal{K} = -Kr^2,\\
\\
&\mathcal{W} = -\int_{-\pi}^{\pi} \int_{-\infty}^{\infty} \omega \rho g(\omega) \mathrm{d}\omega \frac{\bar{\rho}'}{\bar{\rho}} \mathrm{d}\theta.
\end{align}
The difference between Formula ($\ref{entropyderiv}$) and Formula ($\ref{entropyderiv_mean}$) lies entirely in the term $\mathcal{W}$. This tells us that $\bar{\rho}$, unlike $\rho$, even in the absence of diffusion ($D = 0$) can reach a state of stationary entropy provided that
\begin{equation}
\mathcal{K} + \mathcal{W} = 0,
\end{equation} 
which for $K > 0$ would require that $\mathcal{W}$ is positive. We will argue that it is in fact very plausible that $\mathcal{W}$ is positive for a partially synchronized state. In order to see this, we rewrite it in the following form:
\begin{equation}
\mathcal{W} = -\int_{-\pi}^{\pi} \langle \omega \rangle_{\theta} \frac{\bar{\rho}'}{\bar{\rho}} \mathrm{d}\theta,
\end{equation}
where we have identified the factor $\langle \omega \rangle_{\theta}$ as the average natural frequency $\omega$ at angle $\theta$. If we imagine a typical situation we would expect that part of the oscillators with positive frequency become frequency-locked at an angle greater than $\psi$, and vice verca for the oscillators with negative frequency. Moreover, as we move away from $\theta = \psi$ we would expect a declining average density, which implies a negative derivative $\bar{\rho}'$ for positive natural frequencies and a positive derivative $\bar{\rho}'$ for negative natural frequencies. In total, these considerations suggest that $\mathcal{W}$ is indeed likely to be positive for partially synchronized states. Hence, unlike the situation with the total density function $\rho$, entropy considerations pose no a-priori obstacles for $\bar{\rho}$ to converge to a stationary state with an order parameter anywhere in the range $0 \leq r < 1$.

\bibliographystyle{apsrev}
\bibliography{referenser}

\begin{thebibliography}{6}
\expandafter\ifx\csname natexlab\endcsname\relax\def\natexlab#1{#1}\fi
\expandafter\ifx\csname bibnamefont\endcsname\relax
  \def\bibnamefont#1{#1}\fi
\expandafter\ifx\csname bibfnamefont\endcsname\relax
  \def\bibfnamefont#1{#1}\fi
\expandafter\ifx\csname citenamefont\endcsname\relax
  \def\citenamefont#1{#1}\fi
\expandafter\ifx\csname url\endcsname\relax
  \def\url#1{\texttt{#1}}\fi
\expandafter\ifx\csname urlprefix\endcsname\relax\def\urlprefix{URL }\fi
\providecommand{\bibinfo}[2]{#2}
\providecommand{\eprint}[2][]{\url{#2}}

\bibitem[{\citenamefont{Kuramoto}(1984)}]{Kuramoto1984}
\bibinfo{author}{\bibfnamefont{Y.}~\bibnamefont{Kuramoto}},
  \emph{\bibinfo{title}{{Chemical Oscillations, Waves and Turbulence}}}
  (\bibinfo{publisher}{Springer}, \bibinfo{address}{Berlin},
  \bibinfo{year}{1984}).

\bibitem[{\citenamefont{Strogatz}(2000)}]{Strogatz2000}
\bibinfo{author}{\bibfnamefont{S.~H.} \bibnamefont{Strogatz}},
  \bibinfo{journal}{Physica D: Nonlinear Phenomena}
  \textbf{\bibinfo{volume}{143}}, \bibinfo{pages}{1} (\bibinfo{year}{2000}),



\bibitem[{\citenamefont{Mirollo and Strogatz}(1991)}]{Mirollo1991}
\bibinfo{author}{\bibfnamefont{R.~E.} \bibnamefont{Mirollo}} \bibnamefont{and}
  \bibinfo{author}{\bibfnamefont{S.~H.} \bibnamefont{Strogatz}},
  \bibinfo{journal}{J. Statist. Phys.} \textbf{\bibinfo{volume}{63}},
  \bibinfo{pages}{613635} (\bibinfo{year}{1991}).

\bibitem[{\citenamefont{Sakaguchi}(1988)}]{Sakaguchi1988}
\bibinfo{author}{\bibfnamefont{H.}~\bibnamefont{Sakaguchi}},
  \bibinfo{journal}{Progr. Theoret. Phys.} \textbf{\bibinfo{volume}{79}}
  (\bibinfo{year}{1988}).

\bibitem[{\citenamefont{Mirollo and Strogatz}(2007)}]{Mirollo2007}
\bibinfo{author}{\bibfnamefont{R.~E.} \bibnamefont{Mirollo}} \bibnamefont{and}
  \bibinfo{author}{\bibfnamefont{S.~H.} \bibnamefont{Strogatz}},
  \bibinfo{journal}{J. Nonlin. Science} \textbf{\bibinfo{volume}{17}},
  \bibinfo{pages}{309347} (\bibinfo{year}{2007}).

\bibitem[{\citenamefont{Mirollo}(2012)}]{Mirollo2012}
\bibinfo{author}{\bibfnamefont{R.~E.} \bibnamefont{Mirollo}},
  \bibinfo{journal}{Chaos} \textbf{\bibinfo{volume}{22}},
  \bibinfo{pages}{043118} (\bibinfo{year}{2012})

\bibitem[{\citenamefont{Benedetto}(2014)}]{Benedetto2014}
\bibinfo{author}{\bibfnamefont{D.} \bibnamefont{Benedetto}},
\bibinfo{author}{\bibfnamefont{E.} \bibnamefont{Caglioti}}
\bibinfo{author}{\bibfnamefont{U.} \bibnamefont{Montemagno}}
  \bibinfo{journal}{arXiv:1407.6551v1} (\bibinfo{year}{2014})

\bibitem[{\citenamefont{Ott}(2008)}]{Ott2008}
\bibinfo{author}{\bibfnamefont{E.} \bibnamefont{Ott}},
\bibinfo{author}{\bibfnamefont{T. M.} \bibnamefont{Antonsen}},
  \bibinfo{journal}{Chaos} \textbf{\bibinfo{volume}{18}},
  \bibinfo{pages}{037113} (\bibinfo{year}{2008})


\end{thebibliography}

\end{document}